\documentclass[a4paper,11pt,english]{article}
\usepackage[utf8]{inputenc}
\usepackage{babel,amsmath,amssymb,amsthm,enumitem}
\usepackage[sort]{natbib}
\usepackage[colorlinks,allcolors=blue]{hyperref}
\usepackage[margin=1.25in]{geometry}
\usepackage[small,bf]{titlesec}
\titlelabel{\thetitle.\hspace{0.5em}}

\newtheorem{theorem}{Theorem}
\newtheorem{lemma}{Lemma}
\newtheorem{proposition}{Proposition}
\theoremstyle{definition}
\newtheorem{remark}{Remark}
\newtheorem{definition}{Definition}
\newtheorem*{examples}{Examples}
\newtheorem*{assumptions}{Assumptions}
\newcommand{\numer}{num\'eraire}

\DeclareMathOperator{\E}{E}

\DeclareMathOperator{\I}{I}
\DeclareMathOperator{\cl}{cl}
\DeclareMathOperator{\pr}{pr}
\DeclareMathOperator*{\argmax}{argmax}
\newcommand{\s}[2]{\langle #1,#2\rangle}
\renewcommand{\hat}{\widehat}
\renewcommand{\tilde}{\widetilde}
\renewcommand{\epsilon}{\varepsilon}
\renewcommand{\P}{\mathrm{P}}
\newcommand{\F}{\mathcal{F}}
\newcommand{\G}{\mathcal{G}}
\newcommand{\R}{\mathbb{R}}
\newcommand{\Thetab}{\Theta}
\newcommand{\Thetabp}{\Theta'}
\newcommand{\FF}{\mathbb{F}}
\newcommand{\GG}{\mathbb{G}}
\newcommand{\B}{\mathcal{B}}
\renewcommand{\H}{\mathcal{H}}
\renewcommand{\b}[1]{\boldsymbol{#1}}
\newcommand{\bH}{\boldsymbol{h}}
\newcommand{\bC}{\boldsymbol{C}}
\newcommand{\bA}{\boldsymbol{A}}
\newcommand{\bpA}{\boldsymbol{A}^\mathrm{p}}
\newcommand{\pA}{{A}^\mathrm{p}}
\newcommand{\bB}{\boldsymbol{B}}
\newcommand{\bv}{\boldsymbol{v}}

\newcommand{\bha}{\hat{\b\alpha}}
\newcommand{\bhb}{\hat{\b\beta}}
\newcommand{\btY}{\tilde{\b Y}}
\newcommand{\carat}{Carath\'eodory}

\title{Capital growth and survival strategies\\in a market with endogenous
prices}
\author{Mikhail Zhitlukhin\thanks{Steklov Mathematical Institute of the
Russian Academy of Sciences. 8 Gubkina St., Moscow, Russia. Email:
mikhailzh@mi-ras.ru.\newline
The research was supported by the Russian Science Foundation, project no.\ 18-71-10097.}}
\date{24 January 2021}

\begin{document}
\maketitle

\begin{abstract}
We call an investment strategy survival, if an agent who uses it maintains a
non-vanishing share of market wealth over the infinite time horizon. In a
discrete-time multi-agent model with endogenous asset prices determined through a
short-run equilibrium of supply and demand, we show that a survival strategy
can be constructed as follows: an agent should assume that only their
actions determine the prices and use a growth optimal (log-optimal) strategy
with respect to these prices, disregarding the actual prices. Then any
survival strategy turns out to be close to this strategy asymptotically. The
main results are obtained under the assumption that the assets are
short-lived.


\end{abstract}

\section{Introduction}

The main object of study of this paper is asymptotic performance of
investment strategies in stochastic market models. The mathematical theory
of optimal capital growth originated with the works of \citet{Kelly56},
\citet{Latane59}, \cite{Breiman61}, and one of its central results consists
in that an agent who maximizes the expected logarithm of wealth achieves the
fastest asymptotic growth of wealth over the infinite time horizon (see,
e.g., \cite{AlgoetCover88}). The standard assumption made in this theory is
that an agent has negligible impact on a market, and hence asset prices can
be specified by exogenous random processes not depending on agents'
strategies. The aim of this paper is to extend these results and describe
analogues of growth optimal strategies in a multi-agent market model which
may contain assets with endogenously determined prices.

We consider a discrete-time model of a market with two type of assets.
Assets of the first type, further called \emph{exogenous}, have prices and
dividends represented by exogenous random sequences (without loss of
generality, we will assume that the dividends are included in the prices).
Agents get profit or loss when the prices of these assets change. Assets of
the second type, further called \emph{endogenous}, have exogenous dividends,
but their prices are determined endogenously via a short-run equilibrium of
supply and demand. The supply is exogenous, while the demand is generated by
agents' strategies. Typically, an asset with larger dividends is more
attractive and therefore has a higher price. An important simplifying
assumption that will be made in the paper is that the endogenous assets are
short-lived in the sense that they can be though of as financial contracts
which can be bought at some moment of time, yield payments at the next time
instant, and then expire. For example, they can be derivative securities,
loan agreements, contracts for producing goods, etc. It would be interesting
to incorporate long-lived assets (e.g.\ common stock) with endogenous prices
into the model, but this is a more difficult task and is left for future
research.

We are primarily interested in asymptotic behavior of relative wealth of
agents, i.e.\ their shares in total market wealth. We investigate it from a
standpoint of evolutionary dynamics and view a market as a population of
different strategies competing for capital. The central concept of the paper
is the notion of a \emph{survival} strategy. Such a strategy allows an agent
to keep the relative wealth strictly bounded away from zero over the
infinite time horizon. Our goal is to construct a survival strategy in an
explicit form and to find what effect the presence of this strategy has on
the asymptotic distribution of wealth between market agents. In particular,
we are interested in conditions under which a strategy is asymptotically
\emph{dominating}, i.e.\ an agent using it becomes the single survivor in a
market with the relative wealth converging to 1. In order to find a survival
strategy, the notion of a \emph{relative growth optimal} strategy will be
useful. This is a strategy with the logarithm of its relative wealth being a
submartingale. The fact that a non-positive submartingale converges implies
that a relative growth optimal strategy is survival. The convergence of the
compensator of this submartingale allows to obtain a sufficient condition
for a survival strategy to be also dominating.

Note that, in contrast to the optimal growth theory for markets with
exogenous prices, which deals with absolute wealth of agents, we focus on
relative wealth, which turns out to be more amenable to asymptotic analysis
in the case of endogenous prices. \citet[Section~6]{DrokinZhitlukhin20} show
that the goals of maximization of relative and absolute wealth in a model
with endogenous prices may be incompatible.

Our first main result consists in showing that a relative growth optimal
strategy can be constructed as a growth optimal strategy in a market with
exogenous prices equal to the endogenous prices induced by this strategy
when all the agents in the market use it. We find such a strategy in a
tractable form, as a solution of a two-stage optimization problem. On the
first stage, an agent determines the portfolio of exogenous assets by
maximizing the expected log-return (with some adjustments if it is not
integrable); on the second stage the portfolio of endogenous assets is found
via a solution of another maximization problem. We show that this strategy
is relative growth optimal in any strategy profile, irrespectively of the
strategies used by the other agents. Another its feature, which can be
attractive for possible applications, is that it needs to know little
information about the market: only the current total market wealth and the
probability distribution of returns of the exogenous assets and payoffs of
the endogenous assets, but does not require the knowledge of the other
agents’ individual wealth or their strategies. It also does not depend on
the spot prices of the endogenous assets, and so is not affected by the
impact which an agent may have on the market.

Our second main result shows that the obtained strategy becomes the single
surviving strategy in a market if the representative strategy of the other
agents is asymptotically different from it in a certain sense. Consequently,
if some agent uses this strategy, then any other agent who wants to survive
in the market must use an asymptotically similar strategy. As a corollary,
we show that this strategy asymptotically determines the prices of the
endogenous assets.

The results we obtain are tightly related to and generalize the main results
of \citet{AmirEvstigneev+13} and \cite{DrokinZhitlukhin20}. Those papers
also studied survival and growth optimal strategies in markets with
short-lived assets and endogenous prices, however the models were less
general. In the former paper it was assumed that there are only assets with
endogenous prices; the latter paper also included a risk-free bank account
with an exogenous interest rate. Another extension consists in that we allow
the model to include constraints on agents' portfolios specified by random
convex sets. Among other recent papers related to this setting, let us
mention the paper of \cite{BelkovEvstigneev+20}, which builds another model
that includes assets with endogenous prices and a risk-free asset. A
difference with our model is that they assume asset payoffs depend linearly
on the amount of money invested in the risk-free asset, which allows to
reduce the model to previously known results for models without a risk-free
asset.

Let us mention how this paper is related to other results in the literature.
In models with exogenous prices, the asymptotic growth optimality of the
log-optimal strategy (also called the Kelly strategy, after \cite{Kelly56})
was proved for a general discrete-time model by \citet{AlgoetCover88}; a
review of other related results in discrete time can be found in, e.g.,
\citet[Chapter~16]{CoverThomas12} or \cite{HakanssonZiemba95}. For a
treatment of a general model with continuous time and portfolio constraints,
and a connection of growth optimal portfolios (\numer\ portfolios) with
absence of arbitrage, see, e.g., \cite{KaratzasKardaras07}.

Among various lines of research on markets with endogenous prices, our paper
is most closely related to works in evolutionary finance on stability and
survival of investment strategies, which focus on evolutionary dynamics and
properties like survival, extinction, dominance, and how they affect the
structure of a market. Central to this direction are strategies that perform
well irrespectively of competitors' actions. One of the main results
consists in that the strategy which splits its investment budget between
risky assets proportionally to their expected dividends (often also called
the Kelly strategy) survives in a market provided that the agent's beliefs
about the dividends are correct. See, for example, the papers of
\cite{BlumeEasley92,AmirEvstigneev+05,AmirEvstigneev+11,EvstigneevHens+02,Evstigneev+06,HensSchenkHoppe05},
which establish this fact for different models and under different
assumptions. Reviews of this direction can be found in
\cite{EvstigneevHens+16} or \cite{AmirEvstigneev+20}. Typically, the Kelly
strategy turns out to be the only surviving strategy in a market, i.e.\ it
dominates all other asymptotically different strategies. For results on
market wealth evolution when agents use strategies different from the Kelly
strategy, which may result in survival of several strategies, see, e.g.,
\cite{BottazziDindo14,BottazziDindo+18}.

Most of the above mentioned papers (including the present paper) consider
agent-based models, where agents' strategies are specified directly as
functions of a market state. Another large body of literature consists of
results on market selection of investment strategies in the framework of
general equilibrium, where agents maximize utility from consumption. Among
those results one can mention, for example,
\cite{BlumeEasley06,Borovicka20,Sandroni00,Yan08}. \citet{Holtfort19}
provides a detailed survey of the literature in evolutionary finance over
the last three decades, including also some earlier results.

The paper is organized as follows. Section~\ref{sec-model} describes the
model. The main results of the paper are stated in the three theorems
included in Section~\ref{sec-results}. Section~\ref{sec-proofs} contains
their proofs.

\section{The model}
\label{sec-model}

\subsection{Notation}
For vectors $x,y\in \R^N$, we will denote by $\s xy$ their scalar product,
and by $|x| = \sum_n |x^n|$, $\|x\| = \sqrt{\s xx}$ the $L^1$ and $L^2$
norms. If $f\colon \R\to \R$ is a scalar function and $x$ is a vector, then
$f(x) = (f(x^1),\ldots,f(x^N))$ denotes the coordinatewise application of
$f$ to $x$.

By $e$ we will denote the vector consisting of all unit coordinates,
$e=(1,\ldots,1)$, which may be of different dimensions in different
formulas. In particular, $\s ex$ is equal to the sum of coordinates of a
vector $x$.

All equalities and inequalities for random variables are assumed to hold
with probability 1 (almost surely), unless else is stated.

\subsection{Investors and assets}

Let $(\Omega,\F,\P)$ be a probability space with a discrete-time filtration
$\FF=(\F_t)_{t=0}^\infty$ on which all random variables will be defined.
Without loss of generality, we will assume that $\F$ is $\P$-complete and
$\F_0$ contains all $\P$-null events.

The market in the model consists of $M$ agents (investors) and
$N=N_1+N_2$ assets of two types. The assets of the first type are
available in unlimited supply and have exogenous prices; they are treated as
in standard models of mathematical finance. The assets of the second type
are in limited supply; they yield payoffs which are defined exogenously, but
their prices are determined endogenously from an equilibrium of supply and
demand in each time period. These assets are short-lived in the sense that
they can be purchased by agents at time $t$, yield payoffs at $t+1$, and
then get replaced with new assets; agents cannot sell them, and, in
particular, short sales are not allowed (adding short sales would lead to
conceptual difficulties which we prefer to avoid). We will call the assets
of the first and the second type, respectively, exogenous and endogenous.

The prices of the exogenous assets are represented by positive random sequences
$(S_t^n)_{t=0}^\infty$, $n=1,\ldots,N_1$, which are $\FF$-adapted (i.e.\
$S_t^n$ is $\F_t$-measurable). We assume that dividends, if there are any,
are already included in the prices. By $X_t^n=S_t^n/S_{t-1}^n>0$ we will
denote the relative price changes. The payoffs of the endogenous assets (per
one unit of an asset) are represented by non-negative adapted sequences
$(Y_t^n)_{t=1}^\infty$, $n=1,\ldots,N_2$. Without loss of generality, we
assume that the supply of each endogenous asset is equal to 1, so $Y_t^n$ is
the total payoff of an asset. Their prices will be defined later, as we
first need to define agents' strategies, on which they will depend.

The agents enter the market at time $t=0$ with non-random initial wealth
$v_0^m>0$, $m=1,\ldots,M$. Actions of an agent at time $t\ge0$ are
described by a pair of vectors $h_t=(\alpha_t,\beta_t)$, where $\alpha_t \in
\R^{N_1}$, $\beta_t\in \R_+^{N_2}$ specify in what proportions this agent
allocates the current wealth between the assets of the two types (the wealth
sequences are yet to be defined), i.e.\ the proportion $\alpha_t^n$
(respectively, $\beta_t^n$) of wealth is allocated to asset $n$.\footnote{In
the literature, time indices are often shifted by 1 forward (so $h_t$
represents actions at time $t-1$, and, hence, is a predictable sequence).
But in discrete time this is just a matter of notation. For our purposes, it
will be more convenient to let $h_t$ specify actions at time $t$.}

Since $\alpha_t,\beta_t$ are proportions, we require that $\s e{\alpha_t} +
\s e{\beta_t} = 1$. The components of $\beta_t$ are non-negative, because
short sales of the endogenous assets are not allowed. Additionally, we will
assume that it is not possible to buy the endogenous assets on borrowed
funds, i.e.\ $\s e\alpha \in [0,1]$, and hence $\s e\beta\in [0,1]$.
Consequently, $h_t$ assumes values in the set
\[
\H = \{(\alpha,\beta) \in \R^{N_1}\times
\R_+^{N_2} : \s e\alpha \in [0,1],\;  \s e\beta = 1 - \s e\alpha \}.
\]
In order to emphasize that a pair $h_t$ is selected by agent $m$ we will
use the superscript~$m$, e.g.\ $h_t^m=(\alpha_t^m,\beta_t^m)$.

A strategy of an agent consists of investment proportions $h_t^m$
selected at consecutive moments of time. It may (and, usually, does) depend
on a random outcome and market history. In order to specify this dependence,
introduce the measurable space ($\Thetab, \G)$ with
\[
\Thetab = \Omega\times \R_+^M\times (\H^M)^\infty, \qquad
\G = \F\otimes \B(\R_+^M\times (\H^M)^\infty),
\]
where an element $\chi = (\omega, v_0, h_0, h_1,\ldots)\in\Thetab$ consists
of a random outcome $\omega$, a vector of initial wealth
$v_0=(v_0^1,\ldots,v_0^M)\in\R^M_{+}$, and vectors of investment proportions
$h_t = (h_t^1,\ldots,h_t^M)$ selected by the agents at each moment of
time. Let $\GG =(\G_t)_{t\ge 0}$ be the filtration on $\Thetab$ defined by
\[
\G_t = \F_t \otimes\B(\R_+^M\times (\H^M)^{t+1}),
\]
i.e. $\G_t$ is generated by sets $\Gamma\times V\times H_0\times\ldots\times
H_{t}\times (\H^M)^\infty$ with $\Gamma\in\F_t$ and Borel sets
$V\subseteq\R_+^M$, $H_s\subseteq \H^M$. We define a strategy of an agent
as a sequence of $\G_t$-measurable functions
\[
\bH_t(\chi) \colon \Thetab \to \H, \qquad t\ge 0.
\]
Basically, $\bH_t$ can be thought of as a function
$\bH_t(\omega,v_0,h_0,\ldots,h_t)$, but the notation $\bH_t(\chi)$ will be
more convenient for us because we will deal with functions depending on
market histories of different length appearing in one formula, see, e.g.,
\eqref{wealth-equation} below. Note the dependence of $\bH_t$ on the
argument $h_t$, i.e.\ an agent may use information (partial or whole)
about actions of other agents at the same moment of time $t$. This
information may be available to an agent, for example, from asset prices.

We call a vector of initial wealth $v_0$ and a strategy profile
$(\bH^1,\ldots,\bH^M)$ feasible if there exists a sequence of
$\F_t$-measurable functions $h_t(\omega) =
(h_t^1(\omega),\ldots,h_t^M(\omega))\in\H^M$ such that for all $\omega,t,m$
\begin{equation}
\bH_t^m(\chi(\omega)) = h_t^m(\omega),\ \text{where}\ \chi(\omega) =
(\omega,v_0,h_0(\omega),h_1(\omega),\ldots).\label{realization}
\end{equation}
Such a sequence $h(\omega)$ will be called a \emph{realization} of the
agents' strategies corresponding to the given strategy profile and
initial wealth. We do not require the uniqueness of a realization, i.e.\
equation \eqref{realization} may have several solutions. The main results of
the paper will hold for any chosen realization (however, the uniqueness may
be desirable for other applications).

\begin{remark}[On notation]
\label{remark-bold}
By the bold font we denote functions which depend on $\chi$, i.e.\ on a
random outcome and market history, while functions which depend only on a
random outcome $\omega$ (e.g.\ realizations of strategies) are denoted by
the normal font. In particular, if $\b\zeta$ is a function of $\chi$, then,
given a vector of initial wealth and a strategy profile, we denote by
$\zeta(\omega)$ its realization $\b\zeta(\chi(\omega))$, where
$\chi(\omega)$ is as in \eqref{realization}.

If $\xi$ is a random variable, i.e.\ a function of $\omega$ only, we will
sometimes use the same letter to denote the function $\xi(\chi)$ which just
ignores the values of $v_0$ and $h_s$, i.e. $\xi(\chi)=\xi(\omega)$ at an
element $\chi=(\omega,v_0,h_0,h_1,\ldots)$.
\end{remark}

Sufficient conditions for a vector of initial wealth and a strategy profile
to be feasible, in general, can be formulated in terms of assumptions of
fixed-point theorems, but we do not investigate this question in details --
our main goal is to find an optimal strategy, and the strategy which we find
will be optimal in any feasible profile. Nevertheless, it is easy to see
that a simple sufficient condition for the feasibility is that the functions
$\bH_t^m$ do not depend on the argument $h_t$, i.e.\ adapted to the
filtration $\GG^- = (\G_t^-)_{t\ge0}$, where
\[
\G_t^- = \F_t \otimes\B(\R_{+}^M\times (\H^M)^{t}).
\]
This condition can be interpreted as that at each moment of time the
agents decide upon their actions simultaneously and independently of each
other.

Now we can define the prices of the endogenous assets and the wealth
sequences $\bv_t^m(\chi)$ inductively in $t$, beginning with $\bv_0^m(\chi)
= v_0^m$. Denote the prices at time $t$ by $\b p_{t}^n(\chi)$,
$n=1,\ldots,N_2$. Suppose for some $\chi\in\Thetab$ the wealth sequences are
defined up to a moment of time $t$, and $\bv^m_{t}(\chi)\ge 0$ for all $m$.
Then agent $m$ can purchase $\b y_t^{m,n}(\chi)$ units of asset $n$ at
this moment, where
\[
\b y_t^{m,n} = \frac{\beta_t^{m,n} \bv_{t}^m}{\b p_{t}^n},
\]
and $\beta_t^{m,n}$ (also $\alpha_t^{m,n}$ below) are taken from the
component $h_t$ entering $\chi$. In order to clear the market (recall that
the supply of each asset is 1), the prices should be equal to
\begin{equation}
\b p_{t}^n = \sum_{m=1}^M \beta_t^{m,n} \bv_{t}^m.\label{prices}
\end{equation}
Essentially, we employ the principle of moving equilibrium, which operates
with economic variables changing with different speeds. In our model, the
endogenous asset prices move fast, while the investment proportions selected
by the agents move slow; the proportions can be considered fixed while the
prices rapidly adjust to clear the market. The mechanics of this adjustment
process is not important to us (as long as it does not inflict transaction
costs) and it can be modeled by various approaches, e.g.\ limit order books,
auctions, etc. Note that we do not require the agents to agree upon future
asset prices at each random outcome. For a discussion of this moving
equilibrium approach in a similar model, see Section~4
in~\cite{EvstigneevHens+20}.

If $\sum_m \beta_t^{m,n}(\chi)=0$ in formula \eqref{prices} for some $n$,
i.e.\ no one invests in asset $n$, we put $\b y_t^{m,n}(\chi)=0$ for all
$m$; in this case the price $\b p_{t}^n(\chi)$ can be defined in an
arbitrary way with no effect on the agents' wealth, so we will put $\b
p_{t}^n(\chi)=0$ in accordance with \eqref{prices}.

Thus, the portfolio of agent $m$ between moments of time $t$ and $t+1$
consists of $\b y_t^{m,n}$ units of endogenous asset $n$, and $\b x_t^{m,n}$
units of exogenous asset $n$, where
\[
\b x_t^{m,n} = \frac{\alpha_t^{m,n}\bv_{t}^m}{S_{t}^n}.
\]
Consequently, the wealth of this agent at $t+1$ is determined by the
relation
\begin{equation}
\bv_{t+1}^m = \sum_{n=1}^{N_1}\b x_t^{m,n} S_{t+1}^n +
\sum_{n=1}^{N_2} \b y_t^{m,n}Y_{t+1}^n = \biggl(\sum_{n=1}^{N_1}
\alpha_t^{m,n} X_{t+1}^n + \sum_{n=1}^{N_2}
\frac{\beta_t^{m,n}Y_{t+1}^n}{\sum_k
\beta_t^{k,n}\bv_{t}^k}\biggr)\bv_{t}^m
\label{wealth-equation}
\end{equation}
(with $0/0=0$ in the right-hand side).

Observe that in equation \eqref{wealth-equation} the value of $\bv_{t+1}^m$
may become negative, which will make the right-hand side of the equation
meaningless for the next time period. However, below we will introduce
portfolio constraints which prohibit strategies that may lead to negative
wealth. In view of this, we will restrict the domain of the functions $\b
v_t^m$ and define them on sets smaller than $\Thetab$. Namely, introduce
inductively the sets
\[
\Thetab_{t} = \{\chi \in \Thetab : \b
v_{s}^m(\chi)\ge 0\ \text{for all}\ s\le t,\ m=1,\ldots,M\}, \quad t\ge0,
\]
where $\b v_s^m(\chi)$ are computed by \eqref{wealth-equation}. Note that
$\Thetab_0 = \Thetab$, $\Thetab_{t}\supseteq\Thetab_{t+1}$, and $\Thetab_t
\in \G_t^-$. From now on, we will assume that the functions $\b v_t^m$ are
defined only for $\chi\in\Thetab_t$.

It will be also convenient to introduce the sets
\[
\Thetabp_t = \{\chi \in \Thetab_t : \bv_t(\chi)\neq 0\}, \quad t\ge0.
\]
Observe that, essentially, components ($\b\alpha_t,\b\beta_t)$ of an
agent's strategy need to be defined only on $\Thetabp_t$, since elements
from $\Omega\setminus\Thetab_t$ do not correspond to any realization, and on
the set $\{\chi: \bv_t(\chi)=0\}$ they can be defined in an arbitrary way
without any effect on (zero) wealth.

\subsection{Portfolio constraints}
Portfolio constraints in the model are specified by a sequence of
$\G_t^-$-measurable random\footnote{See Section~\ref{sec-random-sets}
for details on random sets.} non-empty closed convex sets
$\bC_t(\chi)\subseteq \H$, $t\ge0$. The constraints are the same for each
agent.

We say that a strategy $\bH$ satisfies the portfolio constraints if
\[
\bH_t(\chi) \in \bC_t(\chi)\ \text{for all $t\ge 0$ and $\chi\in\Thetab$}.
\]
From now on, when writing ``a strategy'', we will always mean a strategy
satisfying the portfolio constraints.

Notice that the sets $\bC_t$ are essentially needed to be defined only for
elements $\chi\in \Thetabp_{t}$. Thus it may be convenient to put, for
example, $\b C_t = \R_+^{N}$ on $\Thetab\setminus\Thetabp_{t}$, without any
effect on realizations of the agents' wealth in the model.

We will consider portfolio constraints only of the following particular
form: they are imposed on the exogenous and endogenous assets separately,
and an agent can freely choose what proportion of wealth to invest in the
assets of each of the two types. Namely, it will be assumed that
\begin{equation}
\label{C-product}
\bC_t = (\bA_t \times \bB_t) \cap \H,
\end{equation}
where $\b A_t$ and $\b B_t$ are $\G_{t}^-$-measurable closed convex sets in
$\R^{N_1}$ and $\R^{N_2}_+$ such that $\s e\alpha \in[0,1]$, $\s e\beta
\in[0,1]$ for any $\alpha\in \bA_t(\chi)$, $\beta\in\bB_t(\chi)$. We also
require that
\begin{align}
&\text{if}\ \alpha \in \bA_t(\chi),\ \text{then}\ \lambda \alpha \in
\bA_t(\chi)\ \text{for any}\ \lambda \in [0,1/\s e\alpha],\label{A-cone}\\
&\text{if}\ \beta \in \bB_t(\chi),\ \text{then}\ \lambda \beta
\in \bB_t(\chi)\ \text{for any}\ \lambda \in [0,1/\s e\beta]\label{B-cone}
\end{align}
(or $\lambda\in[0,\infty)$ if $\s e\alpha=0$ or $\s e\beta = 0$); i.e.\
$\bA_t$ and $\bB_t$ can be represented as intersections of some convex
cones with the sets $\{\alpha \in \R^{N_1}: \s e\alpha \in [0,1]\}$ and
$\{\beta \in \R^{N_2}_+: \s e\beta \in [0,1]\}$ respectively. Note that
relation \eqref{C-product} implies that the sets $\bA_t$, $\bB_t$ cannot
simultaneously (for the same $t,\chi$) consist of only elements $\alpha$ or,
respectively, $\beta$ with zero sum of coordinates, since then the set
$\bC_t$ would be empty.

We will need to further restrict the class of portfolio constraints by
introducing several assumptions on the structure of the sets $\bA_t,\bB_t$.
In what follows, let $K_t(\omega, d\tilde \omega)$ denote some fixed version
of the regular conditional distribution with respect to $\F_t$. By $\P_t$
and $\E_t$ we will denote, respectively, the regular probability and
expectation computed with respect to $K_t$, i.e.\ for a random event
$\Gamma\in\F$ and a random variable $\xi$ we put
\[
\P_t(\Gamma)(\omega) = K_t(\omega, \Gamma),\qquad \E_t(\xi)(\omega) =
\int_{\Omega} \xi(\tilde \omega) K_t(\omega, d\tilde
\omega).
\]
When $\b\xi$ depends also on market history, i.e.\ $\b\xi=\b\xi(\chi)$ is
$\G$-measurable, we put
\[\E_t(\b\xi)(\chi) = \int_{\Omega}
\b\xi(\tilde \omega, v_0, h_0,\ldots) K_t(\omega,
d\tilde \omega),\qquad \chi = (\omega, v_0,h_0,\ldots) ,
\]
provided that the integral is well-defined.

Let us introduce several random sets  which will be needed
to formulate the assumptions on the sets $\bA_t,\bB_t$:
\begin{itemize}[leftmargin=*,itemsep=0em,topsep=0mm]
\item the sets of portfolios of exogenous assets which have non-negative values at the next moment of time:
\[
D_t(\omega) =\{\alpha \in \R^{N_1} : \P_{t}(\s\alpha{X_{t+1}}\ge 0)(\omega) =
1\};
\]

\item the linear spaces of \emph{null investments} (portfolios of exogenous
assets with zero current and next value):
\[
L_t(\omega) = \{\alpha \in \R^{N_1} : \s e\alpha = 0,\;\P_{t}(\s\alpha{X_{t+1}}=0)(\omega)=1\};
\]
\item the projection of $\bA_t$  on the orthogonal space $L_t^\perp$:
\[
\bpA_t(\chi) = \{\alpha \in L_t^\perp(\omega) : \
\exists\, u\in L_t(\omega)\ \text{such that}\ \alpha+u \in \bA_t(\chi)\}.
\]
\end{itemize}
Observe that the sets $D_t$, $L_t$ are $\F_{t}$-measurable, and $\bpA_t$ are
$\G_{t}^-$-measurable. Indeed, we can represent $D_t(\omega) =\{\alpha :
f(\omega,\alpha) =0\}$ with the function $f(\omega,\alpha) =
\E_{t}({\s{\alpha}{X_t}}^-\wedge 1)(\omega)$, which is a \carat\ function,
so $D_t$ is measurable by Filippov's theorem (see
Proposition~\ref{prop-filippov} in Section~\ref{sec-random-sets}). The
set $L_t$ is measurable since it is the intersection of $D_t$, $-D_t$ and
$\{\alpha : \s\alpha e=0\}$. The measurability of $\bpA_t$ follows from
Proposition~\ref{prop-projection}.

\medskip
Now we are ready to formulate the assumptions on the portfolio
constraints. In the remaining part of the paper we always assume that they
are satisfied.
\begin{assumptions}
For all $t\ge 1$ and $\chi=(\omega,v_0,h_0,\ldots)$ it holds that
\begin{enumerate}[label=(A.\arabic*),leftmargin=*,itemsep=0.25em,topsep=0.25em,font=\bfseries]
\item \label{assumption-1} $\bA_t(\chi) \subseteq D_t(\omega)$;
\item \label{assumption-2} there exists $(\alpha,\beta)\in \bC_t(\chi)$
such that $\P_{t}(\s\alpha{X_{t+1}} + \s\beta{Y_{t+1}} > 0)(\omega)=1$;
\item \label{assumption-3} $\bpA_t(\chi)\subseteq
\bA_t(\chi)$;
\item \label{assumption-4} $\bpA_t(\chi)$ is a compact set.
\end{enumerate}
\end{assumptions}

Let us comment on interpretation of these assumptions. \ref{assumption-1}
is imposed to ensure that any strategy which satisfies the portfolio
constraints generates a non-negative wealth sequence. As a consequence, for
the realization of any profile of strategies satisfying the portfolio
constraints we have
\[
\chi(\omega) = (\omega,v_0,h_0(\omega),h_1(\omega),\ldots)\in \Thetab_t\ \text{a.s.\ for
all}\ t\ge 0.
\]
Since the underlying probability space and the filtration are complete, we
can assume that the above inclusion holds for all $\omega\in\Omega$, if
necessary modifying the strategies on a set of zero probability.

Assumption \ref{assumption-2} implies that there exists a strategy with a
strictly positive wealth sequence. Such a strategy can be found via a
standard measurable selection argument, using that $\bC_t$ are measurable
sets. Observe that \ref{assumption-2} is a very mild assumption. For
example, it holds if there is a non-zero vector $\alpha\in \bA_t$ with all
non-negative coordinates (recall that $X_t^n>0$ for all $n$), since then
$(\alpha/|\alpha|,0)\in \bC_t$ by \eqref{A-cone}.

Assumption \ref{assumption-3} means that the agents can remove null
investments from their portfolios. Note that in the literature it is
sometimes required that $L_t\subseteq \bA_t$ (i.e.\ any investment that
leads to no profit or loss is allowed). It is not difficult to see that in
our model this requirement implies \ref{assumption-3}.

Assumption \ref{assumption-4} will allow to reduce the optimal strategy
selection problem to an optimization problem on a compact set. Actually, it
is equivalent to the no arbitrage condition for the exogenous assets -- or,
more precisely, no \emph{unbounded} arbitrage condition -- as we show in the
next section.

\subsection{Absence of unbounded arbitrage opportunities}
Let $U_t(\omega)$ denote the cone of arbitrage opportunities in the
exogenous assets at time $t\ge0$, which consists of all $u\in \R^{N_1}$ such
that
\[
\s eu = 0,\qquad
\P_{t}(\s u{X_{t+1}} \ge 0)(\omega) = 1,\qquad
\P_{t}(\s u{X_{t+1}} > 0)(\omega) > 0.
\]
We say that there are no unbounded arbitrage opportunities in the model if
for all $\chi=(\omega,v_0,h_0,\ldots)\in\Thetab$ and $t\ge 0$ the following
assumption holds:
\begin{enumerate}[label=(A.\arabic*),leftmargin=*,itemsep=0.25em,topsep=0.5em,start=5,font=\bfseries]
\item\label{assumption-na} there is no $u\in U_t(\omega)$ such that $\lambda u \in \bA_t(\chi)$
for any $\lambda>0$.
\end{enumerate}
In other words, an agent cannot infinitely multiply the profit from an
arbitrage opportunity, but the set $\bA_t$ may contain some of them. This
condition is analogous to the no unbounded increasing profit condition
(NUIP), known in connection with \numer\ portfolios, see
\citet[Proposition~3.10]{KaratzasKardaras07}. If there are no constraints on
the exogenous assets (i.e.\ $\bA_t = \{\alpha\in\R^{N_1} : \s e\alpha
\in[0,1]\}$), then \ref{assumption-na} is equivalent to the usual
no-arbitrage condition $U_t = \emptyset$.

\begin{proposition}
\label{prop-na}
Suppose the model satisfies assumptions \ref{assumption-1},
\ref{assumption-3}. Then assumptions \ref{assumption-4} and
\ref{assumption-na} are equivalent.
\end{proposition}
\begin{proof}
It is easy to see that \ref{assumption-4} implies \ref{assumption-na}. Let
us prove the converse implication. Suppose \ref{assumption-na} holds. The
closedness of $\bpA_t$ follows from that $\bA_t$ is closed and assumption
\ref{assumption-3}.

To prove that $\bpA_t$ is bounded, fix $\chi=(\omega,v_0,h_0,\ldots)$ and
suppose, by way of contradiction, that there is a sequence $u_n \in
\bpA_t(\chi)$ such that $|u_n| \to \infty$. The sequence $u_n / |u_n|$ is
bounded, so there exists a convergent subsequence $u_{n_k}/|u_{n_k}| \to u$.
It is easy to see that $\s eu = 0$ (because $\s e{u_n} \in[0,1]$), and
$|u|=1$, $u\in L_t^\perp(\omega)$. The last two properties imply that $u
\notin L_t(\omega)$. Moreover, since $\bA_t \subseteq D_t$, we have
$\P_{t}(\s{u_n}{X_{t+1}} \ge 0)(\omega) = 1$, and hence $\P_{t}(\s
u{X_{t+1}} \ge 0)(\omega) = 1$. Consequently, $u\in U_t(\omega)$.

However, for any $\lambda>0$ and $k$ such that $|u_{n_k}|\ge \lambda$, we
have
\[
\frac\lambda{|u_{n_k}|}
u_{n_k} \in \bA_t(\chi),
\]
and, in the limit, $\lambda u \in \bA_t(\chi)$, so $u$ is an unbounded
arbitrage opportunity, which is a contradiction.
\end{proof}

\begin{examples}
Arbitrage opportunities may be eliminated by imposing appropriate portfolio
constraints, even if the unconstrained model with the same exogenous prices
$S_t$ has arbitrage. As an example, observe that assumption
\ref{assumption-na} automatically holds when portfolio constraints limit
portfolio leverage in the sense that
\begin{equation}
\bA_t \subseteq \{\alpha \in \R^{N_1} : c_t|a^+| \ge |a^-|\} ,\label{leverage}
\end{equation}
where $0\le c_t < 1$ is a random variable (or, in particular, a constant),
and $\alpha^\pm=((\alpha^1)^\pm,\ldots,(\alpha^{N_1})^\pm)$ are the vectors
consisting of the positive and negative parts of the coordinates of
$\alpha$. In this case, if $\s e\alpha = 0$ for $\alpha\in\bA_t$, then
$\alpha=0$, so $U_t\cap \bA_t=\emptyset$ and \ref{assumption-na} holds.

Constraint \eqref{leverage} means that the long positions of a portfolio
should cover the short positions with some margin, which is determined by
the constant $c_t$. If $\alpha\neq0$, this is equivalent to that
\[
|\alpha^+|\ge |\alpha^-|\ \text{and}\ \frac{|\alpha^-|}{|\alpha^+| -|\alpha^-|} \le c_t',
\]
where $c_t' = c_t/(1-c_t)$, which can be interpreted as that the ratio of
the debt to the value of a portfolio (the leverage) is bounded by $c_t'$. If
$c_t=0$, then \eqref{leverage} prohibits short sales of the exogenous
assets. For details on this leverage constraint and how it can be used in
problems of hedging and optimal growth, see e.g.\
\cite{EvstigneevZhitlukhin13,BabaeiEvstigneev+20a,BabaeiEvstigneev+20b}.

Constraint \eqref{leverage} can be relaxed if one requires
\[
\bA_t \subseteq \{\alpha \in \R^{N_1} : d_t + c_t|a^+| \ge |a^-|\},
\]
where $d_t \ge 0$. In this case, $\bA_t$ may include some portfolios with $\s
e\alpha=0$ (besides $\alpha=0$), in particular arbitrage opportunities, but
the set $\bA_t \cap \{\alpha : \s e\alpha = 0\}$ remains bounded, so there
are still no unbounded arbitrage opportunities.
\end{examples}

\section{Main results}
\label{sec-results}
\subsection{The notion of optimality}
We will be interested in long-run behavior of relative wealth of agents,
i.e.\ their shares in total market wealth. We define the total market wealth
and the relative wealth of agent~$m$ as, respectively,
\[
\b W_t = \sum_{m=1}^M \b v_t^m, \qquad \b r_t^m = \frac{\b v_t^m}{\b W_t},
\]
where $\b r_t^m = 0$ if $\b W_t=0$. Recall that $\b v_t$ is defined on the
set $\Thetab_t$, hence we will assume that $\b W_t$ and $\b r_t^m$ are
defined only on this set as well.

For a given feasible strategy profile and a vector of initial wealth, by $W_t(\omega) = \b
W_t(\chi(\omega))$, $r_t^m(\omega) = \b r_t^m(\chi(\omega))$ we will denote
the corresponding realizations defined as in Remark~\ref{remark-bold}. The
realizations of the agents' wealth sequences will be denoted by
$v_t^m(\omega)$.

\begin{definition}
In a feasible strategy profile $(\bH^1,\ldots,\bH^M)$ with initial wealth $v_0\in
\R_+^M$ such that $v_0^m>0$, we call a strategy $\bH^m$
\emph{survival}\footnote{We use the terminology of \cite{AmirEvstigneev+13}.
Note that often a strategy is called survival if
$\limsup\limits_{t\to\infty} r_t^m > 0$, see, e.g., \cite{BlumeEasley92}.}
if
\[
\inf_{t\ge 0} r_t^m >0 \text{ a.s.},
\]
and call it \emph{dominating} if
\[
\lim_{t\to\infty} r_t^m =1 \text{ a.s.}
\]
\end{definition}

Our main goal will be to show that the strategy $\hat\bH$ which we construct
in the next section is survival in any strategy profile and dominating in a
strategy profile if the strategies of the other agents are, in a certain
sense, different from it asymptotically. Consequently, if some agents use
$\hat\bH$, then any other survival strategy should be asymptotically close
to it.

Note that any survival strategy is asymptotically unbeatable in the
following sense: if agent $m$ uses a survival strategy then there exists
a (finite-valued) random variable $\gamma$ such that
\[
r_t^k \le \gamma r_t^m, \qquad k=1,\ldots,M,\; t\ge 0,
\]
which expresses the fact that the wealth of any other agent cannot grow
asymptotically faster than the wealth of an agent who uses a survival
strategy. For a discussion of unbeatable strategies as a game solution
concept in related evolutionary finance models, see e.g.\
\cite{AmirEvstigneev+13}.

At the same time, we would like to emphasize that we do not insist on that
all agents should use only survival strategies, as they may have other
economic goals or make systematic errors. We only investigate what happens
with a market \emph{if} some agents use such strategies.

For construction of a survival strategy, the following notion will be
useful.

\begin{definition}
For a given feasible strategy profile and initial wealth, we call a strategy $\bH^m$
\emph{relative growth optimal} if
\[
v_t^m > 0\ \text{for all $t\ge 0$ and}\ \ln r_t^m\text{ is a submartingale}.
\]
\end{definition}

Since any non-positive submartingale has a finite limit with probability 1
(see, e.g., \citet[Chapter~7.4]{Shiryaev19-2}), for a relative growth
optimal strategy we have $\lim_{t\to\infty} \ln r_t^m > -\infty$, and
therefore $r_\infty^m=\lim_{t\to\infty} r_t^m > 0$. This implies the
following result.

\begin{proposition}
A relative growth optimal strategy is survival.
\end{proposition}

Note that if the relative wealth of an agent is ``infinitesimal'' (so the
strategy of this agent does not affect the prices of the endogenous
assets), then a relative growth optimal strategy for this agent, which
depends on the current endogenous prices $\b p_t$, can be found as a growth
optimal portfolio in a market with $N=N_1+N_2$ exogenous assets, considering
$\b p_t$ as exogenous prices. In particular, if the asset returns are
sufficiently integrable, then such a strategy maximizes the one-period
expected logarithmic return, see, e.g., \cite{AlgoetCover88} or
\citet[Chapter~16]{CoverThomas12}. The important feature of the strategy
that we construct in the next section is that it essentially
depends\footnote{Strictly speaking, this strategy may also depend on some
additional information contained in the market history $\chi$, but only
through the dependence of the portfolio constraints on such information.}
only on the current total market wealth $\b W_t$, but not on the current
endogenous prices, and hence will be a survival strategy for an agent
with any relative wealth.

\subsection{Construction of a relative growth optimal strategy}
In this section we find one relative growth optimal strategy in an explicit
form. The idea behind the construction of this strategy consists in that we
find it as a growth optimal portfolio in a market with endogenous prices
induced by it (see Theorem~\ref{th-numeraire}). We begin with a lemma which
defines the components $\bha,\bhb$ of this strategy. Its statement is
somewhat involved, but clarifying comments will be provided in
Remark~\ref{technical} below.

Recall that we need to define $\bha_t,\bhb_t$ only on the set $\Thetabp_t$,
while on its complement these functions can be defined in an arbitrary way
(respecting the $\G_{t}$-measurability and the portfolio constraints), since
this will not have any effect on realizations of wealth sequences.

In the statement of the lemma and subsequent results, we will use the
following agreement to treat indeterminacies: $0/0 = 0$, $ 0\cdot\ln 0 = 0$,
$ a\cdot\ln 0 = -\infty$ if $a>0$.

\begin{lemma}
\label{lemma-alpha-beta}
The following statements hold true for each $t\ge 0$.

(a) Consider the $\G_{t+1}^-$-measurable vectors $\btY_{t+1}$ in $\R^{N_2}$
with the components
\[
\btY_{t+1}^n(\chi) = Y_{t+1}^n(\omega)\I(\exists\,\beta\in\bB_{t}(\chi) :
\beta^n>0),
\]
and the functions
\[
g_i(x) = \frac 1i + i\arctan\Bigl(\frac xi\Bigr),
\qquad x\in\R_+,\; i=1,2,\ldots
\]
Then there exist $\G_{t}^-$-measurable functions $\bha_{t,i}$ such that for
all $\chi\in\Thetabp_{t}$
\begin{equation}
\bha_{t,i} \in \argmax_{\alpha\in \bpA_t} \bigl\{\E_{t} \ln
g_i(\s\alpha{X_{t+1}}\b W_{t} + |\btY_{t+1}|) - \s e\alpha\bigr\}.\label{max-gi}
\end{equation}

(b) There exists an increasing sequence of $\G_{t}^-$-measurable functions
$i_j(\chi)$, $j\ge 1$, with positive integer values, and a
$\G_{t}^-$-measurable function $\bha_t$ with values in $\bpA_t$, such that
on the set $\Thetabp_{t}$
\[
\bha_t = \lim_{j\to\infty}\bha_{t,i_{j}}.
\]

(c) The set $\tilde \bB_t = \{\beta \in \bB_t(\chi) : |\beta| = 1-\s e{\bha_t(\chi)}\}$ is
non-empty for $\chi\in\Thetabp_{t}$ and
there exists a  $\G_{t}^-$-measurable
function $\bhb_t$ with values in $\bB_t$ such that
for any $\chi\in\Thetabp_{t}$
\begin{equation}
\bhb_t \in \argmax_{\beta\in\tilde \bB_t}
\biggl\{ \E_{t}\frac{\s
{\ln \beta}{\btY_{t+1}}}{\s{\bha_t}{X_{t+1}}\b W_{t} +
|\btY_{t+1}|} \biggr\}.\label{beta-def}
\end{equation}
\end{lemma}

\begin{theorem}
\label{theorem}
In every feasible strategy profile, any strategy $\hat\bH = (\bha,\bhb)$ constructed as in
Lemma~\ref{lemma-alpha-beta} is relative growth optimal.
\end{theorem}

Note that Lemma~\ref{lemma-alpha-beta} defines $\bha_t,\bhb_t$ not
necessarily in a unique way (hence, we write ``any strategy $\hat\bH$'' in
the theorem). This may be so if, for example, some of the vectors $X_t$ have
linearly dependent components.

Let us show that the strategy $\hat\bH$ can be found as an equilibrium
strategy of the representative agent who holds a growth optimal portfolio in
a market with $N_1+N_2$ exogenous assets, where the first $N_1$ assets are
the same as in the original market, and the remaining $N_2$ assets are
treated as exogenous with the prices being equal to the prices of the
endogenous assets induced by $\hat\bH$ in the original market. This notion
of equilibrium is conceptually similar to the one in the Lucas model of
an exchange economy \citep{Lucas78} with the logarithmic utility, though we do
not consider consumption.

Recall that in a market with exogenous prices a strategy with value $\hat
v_t> 0$ is called a growth optimal portfolio (or a \numer\ portfolio) if
for any other strategy with value $v_t\ge0$ it holds that ${v_t}/{\hat v_t}$
is a supermartingale. If $z_t = (v_{t}-v_{t-1})/v_{t-1}$ and $\hat z_t =
(\hat v_t-\hat v_{t-1})/\hat v_{t-1}$ denote the one-period returns on the
strategies' portfolios, then this supermartingality condition is equivalent
to that for each $t\ge0$
\begin{equation}
\label{numeraire-return}
\E_t \frac{1+z_{t+1}}{1+\hat z_{t+1}} \le 1.
\end{equation}

Let $\hat{\b p}_t$ denote the endogenous prices that would clear the market
if all the agents used the strategy $\hat\bH$, i.e.
\[
\hat{\b p}_t^n = \bhb_t^n \b W_t.
\]
Let $\b Z_t$ denote the returns on the endogenous assets in this case,
\[
\b Z_{t+1}^n = \frac{Y_{t+1}^n}{\hat{\b p}_t^n}.
\]
Consequently, the return on a portfolio $(\b\alpha_t, \b\beta_t)$ would be
\begin{equation}
\s{\b\alpha_t}{X_{t+1}} -1 + \s{\b\beta_t}{\b
Z_{t+1}}.\label{portfolio-return}
\end{equation}

\begin{theorem}
\label{th-numeraire}
For any $t\ge 0$, $\chi\in\Thetabp_{t}$, and $(\alpha,\beta)\in \bC_t(\chi)$,
we have $($cf.~\eqref{numeraire-return},~\eqref{portfolio-return}$)$
\begin{equation}
\E_t\frac{\s{\alpha}{X_{t+1}} + \s{\beta}{\b Z_{t+1}}}{\s{\bha_t}{X_{t+1}} +
\s{\bhb_t}{\b Z_{t+1}}} \le 1.\label{numeraire}
\end{equation}
If $\E_t|\ln (\s{\bha_t}{X_{t+1}} + \s{\bhb_t}{\b Z_{t+1}})| < \infty$, then
\begin{equation}
(\bha_t,\bhb_t) \in \argmax_{(\alpha,\beta) \in \bC_t} \E_t \ln
(\s\alpha{X_{t+1}} + \s\beta{\b Z_{t+1}}).\label{numeraire-2}
\end{equation}
\end{theorem}
Relation \eqref{numeraire} expresses the above-mentioned idea of
equilibrium, and relation \eqref{numeraire-2} is an analogue of the
well-known fact that a \numer\ portfolio maximizes one-period expected
log-returns, under the respective integrability condition.

\begin{remark}
\label{technical}
Let us comment on technical aspects of the above results. Why in
Lemma~\ref{lemma-alpha-beta} do we introduce the functions $g_i$ and
consider maximization problem \eqref{max-gi}? Actually, we would like to
find a strategy $(\bha,\bhb)$ such that
\begin{equation}
\label{essence-alpha}
\text{$\bha_t$ maximizes}\ \E_{t} \ln (\s\alpha{X_{t+1}}\b W_{t} + |Y_{t+1}|) -
\s e\alpha\ \text{over $\alpha\in \bA_t$},
\end{equation}
and, for this $\bha_t$, to define the component $\bhb_t$ as in
\eqref{beta-def}. This strategy would satisfy inequalities
\eqref{alpha-ineq} and \eqref{beta-ineq}, which play the key role in the
proofs.

But it may be not possible to define $\bha_t$ in this way, since problem
\eqref{essence-alpha} may have no solution. For this reason, we find the
solutions $\bha_{t,i}$ of the maximization problems truncated by the
functions $g_i$ and select a convergent subsequence. Then inequalities
\eqref{alpha-ineq},~\eqref{beta-ineq} still remain satisfied. To ensure that
such a subsequence exists, we use the observation that it is possible to
maximize not over the whole set $\bA_t$ but over its compact subset
$\bpA_t$. We also replace $\b Y_t$ with $\btY_t$ to avoid the situation when
an asset yields a positive payoff with positive conditional probability, but
it is not possible to invest in it.
\end{remark}

Note that when no portfolio constraints are imposed on the endogenous
assets, i.e.\ $\bB_t =\{\beta\in \R_{+}^{N_2} : |\beta|\le1\}$ and hence
$\btY_{t+1} = Y_{t+1}$, we can find $\bhb_t$ explicitly:
\begin{equation}
\bhb_t^n = \E_{t} \frac{Y_{t+1}^n}{\s{\bha_t}{X_{t+1}}\b W_{t} +
|Y_{t+1}|}\label{beta-explicit}
\end{equation}
(this formula will be used in Section~\ref{section-dominance}). Indeed, for
$\bhb_t$ defined by \eqref{beta-explicit}, we have $|\bhb_t| = 1-
\s{e}{\bha_t}$ as follows from equality \eqref{alpha-eq} below, and for any
$\beta\in\R_+^{N_2}$ with $|\beta| = 1- \s{e}{\bha_t}$ we have
\[
\E_{t}\frac{\s {\ln \beta}{Y_{t+1}}}{\s{\bha_t}{X_{t+1}}\b W_{t}
+ |Y_{t+1}|} = \s {\ln \beta}{\bhb_t} \le  \s {\ln\bhb_t}{\bhb_t},
\]
so $\bhb_t$ indeed delivers the maximum in \eqref{beta-def}. The inequality
here follows from Gibb's inequality (see Proposition~\ref{lemma-ineqs} below).

To conclude this section, let us show how the above theorems generalize
known results on asymptotically optimal strategies. An immediate corollary
from Theorem~\ref{th-numeraire} is that in a market with only exogenous
assets the strategy $\hat\alpha_t$ is a \numer\ portfolio. Note that in
this case $\hat\alpha_t$ depends only on $\omega$, but not on market
history (assuming that the constraints set $\bC_t$ also depend only on
$\omega$).

In a market with only endogenous assets and no portfolio constraints, as
follows from \eqref{beta-explicit}, the optimal strategy is given by
\[
\hat\beta_t^n = \E_{t}\frac{Y_{t+1}^n}{|Y_{t+1}|}
\]
(note that again we have the dependence on $\omega$ only). This strategy was
obtained by \cite{AmirEvstigneev+13}; see also the earlier results of
\cite{EvstigneevHens+02,AmirEvstigneev+05,HensSchenkHoppe05} for models with
short-lived assets which impose additional assumptions on admissible
strategies or on asset payoffs.

Finally, suppose that there is only one exogenous asset, short sales of this
asset are not allowed, and there are no other portfolio constraints, i.e.\
$\bC_t = \R_+\times \R_+^{N_2}$. Then $\bha_t(\chi)$ is defined as follows:
if $\E_{t} (X_{t+1} \b W_{t} / |Y_{t+1}|) \le 1$, then $\bha_t =0$;
otherwise $\bha_t$ is the unique solution of the equation
\[
\E_{t} \frac{X_{t+1}\b W_{t} }{ \alpha + |Y_{t+1}|} = 1.
\]
This can be seen from relations \eqref{alpha-ineq}--\eqref{alpha-eq} below.
Indeed, if $\E_{t} (X_{t+1}\b W_{t} / |Y_{t+1}|) \le 1$, then equality
\eqref{alpha-eq} can be true only if $\bha_t=0$. In the case $\E_{t}
(X_{t+1}\b W_{t} / |Y_{t+1}|) > 1$, equality \eqref{alpha-eq} has two
solutions, the zero one and a non-zero one. But if $\bha_t=0$, then
\eqref{alpha-ineq} cannot hold true for $\alpha>0$, hence we are left only
with the non-zero solution.

After $\bha_t$ has been defined as above, the component $\bhb_t$ can be
found from \eqref{beta-explicit}, which gives
\[
\bhb_t^n = \E_{t} \frac{Y_{t+1}^n}{ \bha_t X_{t+1}\b W_{t} + |Y_{t+1}|}.
\]
This strategy was obtained by \cite{DrokinZhitlukhin20} in the case when the
sequence $X_t$ is predictable (e.g.\ the exogenous asset is a risk-free bond
or cash); see \cite{Zhitlukhin19,Zhitlukhin20b} for its extensions to
continuous time.

\subsection{Asymptotic proximity of survival strategies}
\label{section-dominance}
In this section we investigate evolution of relative wealth of strategies
different from $\hat\bH$. The theorem below will be stated for the case when
there are no portfolio constraints on the endogenous assets ($\bB_t
=\{\beta\in \R_{+}^{N_2} : |\beta|\le1\}$). This assumption is necessary
because the proof relies on the explicit form of $\bhb_t$ given
by~\eqref{beta-explicit}.

Given a feasible strategy profile and a vector of initial wealth, by $\bar h=(\bar
\alpha,\bar \beta)$ we will denote the realization of the representative
strategy of all the agents, which we define as the weighted sum of their
strategies with $r_t^m$ as the weights:
\[
\bar \alpha_t = \sum_{m=1}^M r_t^m  \alpha_t^m, \qquad
\bar \beta_t = \sum_{m=1}^M r_t^m  \beta_t^m,
\]
where $\alpha_t,\beta_t,r_t$ are the corresponding realizations. In a
similar way, by $\tilde h=(\tilde \alpha,\tilde \beta)$ we will denote the
realization of the representative strategy of agents $m=2,\ldots,M$
weighted with their relative wealths excluding agent 1:
\[
\tilde \alpha_t = \sum_{m=2}^M \frac{r_t^m}{1-r_t^1}  \alpha_t^m, \qquad
\tilde \beta_t = \sum_{m=2}^M \frac{r_t^m}{1-r_t^1}  \beta_t^m,
\]
where $0/0=0$.

\begin{theorem}
\label{th-dominance}
Suppose $\bB_t =\{\beta\in \R_{+}^{N_2} : |\beta|\le1\}$, and agent 1
uses the strategy $\bH^1=\hat\bH$. Considering the realizations of the
strategies, the wealth sequences, and the constraints sets, let
\[
Q_{t+1}(\omega) = \max_{\alpha\in A_{t}(\omega)} \s{\alpha}{X_{t+1}(\omega)} +
\frac{| Y_{t+1}(\omega)|}{W_{t}(\omega)}.
\]
Then, with probability 1,
\begin{equation}
\label{proximity}
\sum_{t=0}^\infty \biggl(\frac{\s{\alpha_{t}^1 -
\bar\alpha_{t}}{X_{t+1}}}{Q_{t+1}}\biggr)^2 + \|\beta_{ t}^1 -
\bar\beta_{t}\|^2 < \infty,
\end{equation}
and
\begin{equation}
\lim_{t\to\infty}r^1_t= 1\ \text{on the set}\ \biggl\{\omega: \sum_{t=0}^\infty \biggl(\frac{\s{\alpha_{t}^1 -
\tilde\alpha_{t}}{X_{t+1}}}{Q_{t+1}}\biggr)^2 + \|\beta_{ t}^1 -   \tilde\beta_{t}\|^2 =\infty\biggr\}.
\label{dominance}
\end{equation}
\end{theorem}
Note that the maximum in the definition of $Q_{t+1}$ is attained because,
according to Proposition~\ref{prop-na}, it can be taken over the
compact set $\pA_t(\omega)$. Furthermore, $Q_{t+1}>0$ by
assumption~\ref{assumption-2}.

Relation \eqref{proximity} essentially shows that if one agent uses the
strategy $\hat\bH$ then this agent asymptotically determines the
representative strategy of the market so that $\bar h_t$ becomes close to
$\hat h_t$ in the sense that the series in \eqref{proximity} converges, and,
consequently,
\[
\frac{\s{\alpha_{t}^1 -
\bar\alpha_{t}}{X_{t+1}}}{Q_{t+1}} \to 0, \quad \beta_t^1 - \bar\beta_t \to
0\quad
\ \text{as $t\to\infty$}.
\]

Relation \eqref{dominance} provides a sufficient condition for an agent
using the strategy $\hat\bH$ to dominate in the market, which happens when
the realization of the representative strategy of the other agents is
asymptotically different from the realization of $\hat\bH$ in the sense that
the series in \eqref{dominance} diverges. From here, we also get a necessary
condition for a strategy to be survival. Indeed, a survival strategy must
survive against $\hat\bH$, so if agents $m=1,\ldots,M-1$ use the 
strategy $\hat\bH$, and hence can be considered as a single agent, the
remaining agent $m=M$ has to use a strategy with a realization close to
$\hat h$ in the sense of \eqref{dominance}.

Another corollary from Theorem~\ref{th-dominance} is that the presence of an
agent who uses the strategy $\hat\bH$ asymptotically determines the
relative prices $\rho_t^n = {p_t^n}/{W_t}$ of the endogenous assets. It is
not difficult to see that $\rho_t^n = \bar \beta_t^n$, and hence
\eqref{proximity} implies that for each $n$ and $t\to\infty$ we have
$\hat\beta_t^n - \rho_t^n \to 0$.

\section{Proofs of the main results}
\label{sec-proofs}
\subsection{Auxiliary results on random sets}
\label{sec-random-sets}
In this section we provide several results from the theory of random sets
which will be used in the proofs for dealing with portfolio constraints.

By a \emph{random set} (or a \emph{measurable correspondence}) in $\R^N$
defined on a measurable space $(S,\mathcal{S})$ we call a set-valued
function $\phi\colon S \to 2^{\R^N}$ such that for any open set $A\subseteq
\R^N$ it holds that $\phi^{-1}(A) \in \mathcal{S}$, where $\phi^{-1}(A) =
\{s: \phi(s) \cap A \neq\emptyset\}$ is the lower inverse of $A$. An
equivalent definition is that the distance function $d(x,\phi(s))$ is
$\mathcal{S}$-measurable for any $x\in\R^N$ (where $d(x,\emptyset)=\infty$).
In what follows, the role of $(S,\mathcal{S})$ will be played by
$(\Omega,\F_t)$, $(\Thetab, \G_t)$, or $(\Thetab, \G_t^-)$.

A random set is called closed (respectively, compact, non-empty) if
$\phi(s)$ is closed (compact, non-empty) for any $s\in\mathcal{S}$. A
\emph{measurable selector} is an $\mathcal{S}$-measurable function $\xi$
such that $\xi(s)\in \phi(s)$ for any $s$. A function $f(s,x)\colon S\times
\R^N\to\R$ is called a \emph{\carat\ function} if it is measurable in $s$
and continuous in $x$.

The following results are known for random sets in $\R^N$.

\begin{proposition}
If $\phi_n$, $n=1,2,\ldots$\,, are random sets, then $\cup_n \phi_n$ is a
random set; if $\phi_n$ are also closed, then $\cap_n\phi_n$ is a closed
random set.
\end{proposition}

\begin{proposition}[Filippov's theorem]
\label{prop-filippov}
Suppose $\phi$ is a non-empty compact random set, $f$ is a \carat\ function,
and $\pi$ is a measurable function. Then the correspondence
\[
\psi(s) = \{x\in \phi(s) : f(s,x)=\pi(s)\}
\]
is measurable and compact. Moreover, if $\psi$ is non-empty, then it has a
measurable selector $\xi$, and hence $f(s,\xi(s)) = \pi(s)$.
\end{proposition}

\begin{proposition}[Measurable maximum theorem]
For a non-empty compact random set $\phi$ and a \carat\ function $f$, let
$\mu$ be the maximum function and $\psi$ be the argmax correspondence
defined by
\[
\mu(s) = \max_{x\in\phi(s)} f(s,x), \qquad
\psi(s) = \argmax_{x\in \phi(s)} f(s,x).
\]
Then $\mu$ is measurable, and $\psi$ is non-empty, compact, measurable, and
has a measurable selector.
\end{proposition}

Proofs of the above results can be found in the book of
\citet[Chapter~18]{AliprantisBorder06} for random sets in general metric
spaces, except the result about $\cap_n \phi_n$, which holds (in a metric
space) if $\phi_n$ are compact. For $\R^N$, it can be extended to closed
sets using that $\R^N$ is $\sigma$-compact.

For the reader's convenience, the following results are provided with
proofs (they are not included in the above-mentioned book).

\begin{proposition}
\label{prop-projection}
Let $L$ be a random linear subspace of $\R^N$ (i.e.\ for each $s$ the set
$L(s)$ is a linear space and the correspondence $L$ is measurable),
$L^\perp$ be the orthogonal space, and $\phi$ be a closed random set in
$\R^N$. Then the projection correspondence
\[
\pr_L \phi(s) = \{x \in L(s) : \exists\, y \in L^\perp(s)\ \text{such that}\
x+y \in \phi(s)\}
\]
is measurable.
\end{proposition}
\begin{proof}
By Castaing's theorem (see Corollary 18.14 in \cite{AliprantisBorder06}), a
non-empty closed correspondence is measurable if and only if it can be
represented as the closure of a countable family of measurable selectors
from it. Hence, we can find measurable $\xi_i$ such that $\phi(s) =
\cl\{\xi_i(s), i\ge1\}$ on the set $\{s : \phi(s)\neq \emptyset\}$. Using
that
\[
\cl(\pr_L\phi(s)) =
\begin{cases}
\cl\{\pr_L \xi_i(s), i\ge 1\}, &\text{if}\ \phi(s)\neq \emptyset,\\
\emptyset,&\text{if}\ \phi(s)= \emptyset,\\
\end{cases}
\]
one can see that $\cl(\pr_L\phi)$ is measurable. Since the measurability of
a correspondence is equivalent to the measurability of its closure
\cite[Lemma~18.3]{AliprantisBorder06}, $\pr_L\phi$ is measurable.
\end{proof}

\begin{proposition}
\label{convergent-subsequence}
Let $\phi$ be a non-empty compact random set and $\xi_n$ be a sequence of
measurable selectors from it. Then there exists a measurable selector $\xi$
from $\phi$ and a sequence of measurable functions $1\le
i_1(s)<i_2(s)<\ldots$ with integer values such that
$\lim_{j\to\infty}\xi_{i_j(s)}(s) = \xi(s)$ for all $s$.
\end{proposition}
\begin{proof}
The set $\psi(s) = \cap_{n} \mathrm{cl}\{\xi_k(s),\; k\ge n\}$ is
measurable, non-empty, and closed, so there exists a measurable selector
$\xi \in \psi$ (by Castaing's theorem mentioned above). Then the sequence
$i_j$ can be constructed by induction as follows. Put $i_1=1$. If $i_j$ is
defined, consider the random set $\eta_j(s) = \{k> i_j(s) : |\xi_k(s) -
\xi(s)| \le j^{-1}\} \subset \mathbb{N}$, which is measurable, non-empty,
and closed. Let $i_{j+1}$ be a measurable selector from $\eta_j$. Then
$|\xi_{i_{j+1}} - \xi| < j^{-1}$, which gives the desired convergence.
\end{proof}

\subsection{Proof of Lemma~\ref{lemma-alpha-beta}}
\textit{Proof of claim (a)}. Fix any $t\ge 0$. Let $f_i(\chi,\alpha)$ be the
function which is maximized in the definition of $\bha_{t,i}$, i.e.\ on the
set $\Thetabp_{t}$ put
\[
f_i = \E_{t} \ln g_i( \s\alpha{X_{t+1}}\b W_{t} + |\btY_{t+1}|) - \s e\alpha,
\]
while on the set $\Thetab\setminus\Thetabp_{t}$ put $f_i=0$. The function
$f_i$ is a \carat\ function, and the set $\bpA_t$, over which it is
maximized, is compact by Proposition~\ref{prop-na}. Hence the
measurable maximum theorem implies the existence of a measurable selector
$\bha_{t,i}$ from the $\argmax$ in \eqref{max-gi}.

\textit{Proof of claim (b)} readily follows from
Proposition~\ref{convergent-subsequence}. Before we continue with the proof
of claim (c), let us show that $\bha_t$ satisfies a number of relations that
will be used in its proof, as well as in the proof of Theorem~\ref{theorem}.

\begin{lemma}
For any $t\ge 0$, $\chi\in\Thetabp_{t}$, and $\alpha \in \bA_t(\chi)$ we have
\begin{align}
&\P_{t}(\s{\bha_t}{X_{t+1}}\b W_{t}
+ |\btY_{t+1}| > 0)=1,\label{ineq-positive}\\
&\E_{t} \biggl(\frac{\s{\bha_t - \alpha}{X_{t+1}}\b W_{t}}
{\s{\bha_t}{X_{t+1}}\b W_{t}
+ |\btY_{t+1}|} \biggr)\ge \s{e}{\bha_t - \alpha},\label{alpha-ineq}\\
&\E_{t} \biggl(\frac{\s{\bha_t}{X_{t+1}}\b W_{t}} {\s{\bha_t}{X_{t+1}}\b W_{t} + |\btY_{t+1}|} \biggr) =
\s{e}{\bha_t}  .\label{alpha-eq}
\end{align}
\end{lemma}
\begin{proof}
Fix $t,\chi,\alpha,i$, and consider the function $u(\epsilon) =
f_i((1-\epsilon)\bha_t + \epsilon \alpha)$, $\epsilon\in[0,1]$, i.e.\
\begin{equation}
u(\epsilon) = \E_{t} \ln g_i\bigl( \s{(1-\epsilon) \bha_{t,i} + \epsilon
\alpha}{X_{t+1}}\b W_t +
|\btY_{t+1}|\bigr) (\chi)  - \s{e}{(1-\epsilon)\bha_{t,i}(\chi) + \epsilon\alpha}.\label{u-eps}
\end{equation}
Since $\ln g_i(x)$ is concave for $x\ge 0$, the function $u(\epsilon)$ is
also concave. As it attains the maximum value at $\epsilon=0$, the right
derivative $u'(0)\le 0$. We want to interchange the order of differentiation
and taking the expectation. The expectation in \eqref{u-eps} can be written
as $\E_t \ln g_i(q(\tilde \omega,\epsilon)) = \int_\Omega \ln g_i(q(\tilde
\omega, \epsilon)) K_t(\omega,d\tilde\omega)$ with
\[
q(\tilde\omega, \epsilon) = \s{(1-\epsilon) \bha_{t,i}(\tilde \chi) + \epsilon
\alpha}{X_{t+1}(\tilde \omega)}\b
W_t(\tilde \chi) +
|\btY_{t+1}(\tilde\chi)|,
\]
where $\tilde \chi = (\tilde\omega, v_0,h_0,\ldots)$ and $v_0,(h_s)_{s\ge
0}$ are taken from $\chi=(\omega,v_0,h_0,\ldots)$. By applying Fatou's
lemma, we obtain ($\tilde \omega$ is omitted for brevity)
\begin{equation}
(\E_t \ln g_i(q(\epsilon)))'_{\epsilon=0} \ge \E_t \frac{g'_i(q(
0))}{g_i(q(0))} \s{\alpha-\bha_{t,i}}{X_{t+1}} \b W_t.\label{u-fatou}
\end{equation}
Fatou's lemma can be applied since for $\epsilon\in[0,1)$ we have the lower
bound ($\P_t$-a.s.\ in $\tilde \omega$)
\begin{multline*}
\frac{\ln g_i(q(\epsilon)) - \ln
g_i(q(0))}{\epsilon} \ge (\ln g_i(q(\epsilon)))'
= \frac{g_i'(q(\epsilon))}{g_i(q(\epsilon))}\s{\alpha-\bha_{t,i}}{X_{t+1}}\b W_t \\
\ge -i g_i'((1-\epsilon)\s{\bha_{t,i}}{X_{t+1}}\b W_t)\s{\bha_{t,i}}{X_{t+1}}\b W_t
\ge  -\frac{i^3}{1-\epsilon}.
\end{multline*}
Here in the first inequality we used the concavity of $\ln
g_i(q(\epsilon))$. In the second inequality we used the relation
$\P_t(\s\alpha{X_{t+1}} \ge 0)=1$, the bound $g_i(x)\ge 1/i$, and that
$g'_i(x)$ is non-increasing for $x\ge 0$. The last last inequality holds
because $g_i'(x)x \le i^2$.

Therefore, from \eqref{u-eps} and \eqref{u-fatou}  we obtain 
\begin{equation}
0 \ge u'(0) \ge    \E_{t}(\b\xi_i
\s{\alpha}{X_{t+1}}\b W_{t}) - \E_{t}(\b\xi_i
\s{\bha_{t,i}}{X_{t+1}}\b W_{t}) - \s{e}{\alpha-\bha_{t,i}},
\label{f-ineq}
\end{equation}
where
\[
\b\xi_i = \frac{g_i'(q(0))}{g_i(q(0))} =
\frac{g_i'(\s{\bha_{t,i}}{X_{t+1}}\b W_{t}+ |\btY_{t+1}|)}
{g_i(\s{\bha_{t,i}}{X_{t+1}}\b W_{t} + |\btY_{t+1}|)}.
\]
One can see that for all $x\ge 0$
\begin{equation}
0\le\frac{xg_i'(x)}{g_i(x)} \le 2, \qquad
\lim_{i\to\infty}\frac{xg_i'(x)}{g_i(x)}= \I(x>0).\label{g-limit}
\end{equation}
The above inequality can be obtained by using that $\arctan(x/i) \ge x/(2i)$
if $x\le i$ and $\arctan(x/i) \ge \pi/4$ if $x\ge i$; the computation of the
limit is straightforward. Relations \eqref{g-limit} allow to apply the
dominated convergence theorem to the second expectation in \eqref{f-ineq},
which gives
\begin{equation}
\lim_{i\to\infty}\E_{t}(\b\xi_i \s{\bha_{t,i}}{X_{t+1}}\b W_{t}) =
\E_{t}\biggl(\frac{\s{\bha_t}{X_{t+1}}\b W_{t}} {\s{\bha_t}{X_{t+1}}\b W_{t} +
|\btY_{t+1}|} \biggr),\label{limit-1}
\end{equation}
where at this point we assume $0/0=0$ in the right-hand side (according
with the indicator in \eqref{g-limit}). However, as follows from assumption
\ref{assumption-2}, there exists $\tilde\alpha$ such that
$\P_{t}(\s{\tilde \alpha}{X_{t+1}}\b W_{t} + |\btY_{t+1}| > 0)=1$. Applying
Fatou's lemma to the first expectation in \eqref{f-ineq} with $
\alpha=\tilde\alpha$, we find that \eqref{ineq-positive} must hold, since
otherwise we would have
\[
\liminf_{i\to\infty}\E_{t}(\b\xi_i \s{\tilde\alpha}{X_{t+1}}\b W_{t}) = +\infty,
\]
which contradicts \eqref{f-ineq}. Consequently, for any $\alpha\in \bA_t$ we
have
\[
\liminf_{i\to\infty} \E_{t}(\b\xi_i
\s{\alpha}{X_{t+1}}\b W_{t}) \ge
\E_{t}\biggl(\frac{\s{\alpha}{X_{t+1}}\b W_{t}}
{\s{\bha_t}{X_{t+1}}\b W_{t}
+ |\btY_{t+1}|} \biggr),
\]
which together with \eqref{f-ineq} and \eqref{limit-1} implies
\eqref{alpha-ineq}.

Let us prove \eqref{alpha-eq}. If $\s e{\bha_t(\chi)}=1$, it clearly follows
from \eqref{alpha-ineq} with $\alpha=0$. If $\s e{\bha_t(\chi)}<1$, we can
consider small $\epsilon> 0$ and take as $\alpha$ in \eqref{alpha-ineq}
\[
\alpha^{(\pm\epsilon)} := (1\pm\epsilon)\bha_t(\chi) \in  \bA_t(\chi),
\]
which gives \eqref{alpha-eq} after simple transformations.
\end{proof}

\medskip
\noindent
\textit{Proof of claim (c) of Lemma~\ref{lemma-alpha-beta}}. 
If $\bB_t(\chi) \neq \{0\}$, then $\tilde\bB_t(\chi)\neq\emptyset$ in view
of \eqref{B-cone}. If $\bB_t(\chi) = \{0\}$, then $\btY_t(\chi) = 0$, and
\eqref{alpha-eq} implies that $\s e{\bha_t(\chi)}=1$, so $\tilde\bB_t(\chi)
= \{0\}$ is non-empty again.

Let $f(\chi,\beta)$ denote the function being maximized in \eqref{beta-def}:
\[
f(\chi,\beta) = \sum_{n=1}^{N_2} \ln\beta^n \E_{t}\biggl( \frac{\btY_{t+1}^n}{\s
{\bha_t}{X_{t+1}}\b W_{t} + |\btY_{t+1}|}\biggr)(\chi).
\]
The function $f$ may be discontinuous in $\beta$. In order to apply the
measurable maximum theorem, let us take $\G_{t-1}$-measurable
$\tilde{\b\beta}(\chi)\in \bB_t(\chi)$ such that $|\tilde{\b\beta}(\chi)| =
1-\bha_t(\chi)$ and $\tilde{\b\beta}^n(\chi)>0$ if
$\P_{t}(\btY_{t+1}^n>0)(\omega)>0$. Then we can consider the function
\[
\tilde f(\chi,\beta) = \max (f(\chi,\beta), f(\chi,\tilde{\b\beta}(\chi))),
\]
which is a \carat\ function and satisfies the relation
\[
\argmax_{\beta \in \tilde\bB_t} f(\chi,\beta) =
\argmax_{\beta \in \tilde\bB_t} \tilde f(\chi,\beta).
\]
Hence the measurable maximum theorem can be applied to $\tilde f$, giving
$\bhb_t$ which also maximizes $f$.

\subsection{Proofs of Theorems \ref{theorem} and \ref{th-numeraire}}
Let us prove two more inequalities which together with \eqref{alpha-ineq}
will be used in the proofs.
\begin{lemma}
For any $t\ge 0$, $\chi\in\Thetabp_{t}$, and $\beta\in \bB_t(\chi)$ we have
\begin{gather}
\E_{t}\biggl(\frac{\s{\ln \bhb_t -\ln\beta}{\btY_{t+1}}}{\s{\bha_t}{X_{t+1}}\b W_{t}
+ |\btY_{t+1}|}\biggr) \ge  |\bhb_t| - |\beta|,\label{beta-ineq}\\
\E_t \frac{|\btY_{t+1}| - \sum_n
\beta^n \btY_{t+1}^n/\bhb_t^n}{\s{\bha_t}{X_{t+1}}\b W_t + |\btY_{t+1}|} \ge
|\bhb_t| - |\beta|,\label{corollary}
\end{gather}
where in \eqref{corollary} we let $\beta^n \btY_{t+1}^n(\chi)/\bhb_t^n(\chi)
= 0$ if $\bhb_t^n(\chi)=0$ $($then $\P_t(\btY_{t+1}^n=0)(\chi)=1$ as follows
from~\eqref{beta-def}{}$)$.
\end{lemma}

\begin{proof}
Clearly, \eqref{beta-ineq} holds if $|\beta|=|\bhb_t(\chi)|$, as follows
from the definition of $\bhb_t$. If $|\beta|\neq|\bhb_t(\chi)|$, we have
\begin{multline*}
\E_{t}\biggl(\frac{\s{\ln \bhb_t -\ln\beta}{\btY_{t+1}}}{\s{\bha_t}{X_{t+1}}\b W_{t} + |\btY_{t+1}|}\biggr) \ge
\E_{t}\biggl(\frac{ |\btY_{t+1}| \ln({|\bhb_t|}/{|\beta|})}{\s{\bha_t}{X_{t+1}}\b W_{t} +
|\btY_{t+1}|}\biggr) \\ \ge \E_{t}\biggl(\frac{ |\btY_{t+1}|}{\s{\bha_t}{X_{t+1}}\b W_{t} +
|\btY_{t+1}|}\biggr) \frac{|\bhb_t|-|\beta|}{|\bhb_t|} = |\bhb_t| - |\beta|,
\end{multline*}
where in the first inequality we represented $\ln \beta = \ln
(\beta|\bhb_t|/|\beta|) - \ln(|\bhb_t|/|\beta|)$ and applied
\eqref{beta-ineq} to $\beta|\bhb_t|/|\beta|$ instead of $\beta$; in the
second inequality we used the estimate $\ln a \ge 1-a^{-1}$; and in the
equality applied \eqref{alpha-eq}. This proves \eqref{beta-ineq}.

To prove \eqref{corollary}, observe that the function
\[
f(\epsilon) =
\E_t \biggl(\frac{\s{\ln ((1-\epsilon)\bhb_t + \epsilon\beta)}{\btY_{t+1}}}
{\s{\bha_t}{X_{t+1}}\b W_{t} + |\btY_{t+1}|} \biggr)
-
|(1-\epsilon)\bhb_t + \epsilon\beta_t|, \qquad \epsilon\in[0,1],
\]
attains its maximum at $\epsilon=0$ and is differentiable on $[0,1)$, so its
derivative at zero
\[
f'(0) = \sum_{n=1}^{N_2} \frac{(\beta^n - \bhb_t^n)\btY_{t+1}^n}
{\bhb_t^n(\s{\bha_t}{X_{t+1}}\b W_{t} + |\btY_{t+1}|)} + |\bhb_t|
- |\beta| 
\]
should be non-positive, which gives \eqref{corollary} (here, the $n$-th term
in the sum is treated as zero when $\bhb_t^n(\chi)=0$, and hence
$\btY_{t+1}^n=0$).
\end{proof}

\begin{proof}[Proof of Theorem~\ref{theorem}]
Assume that the strategy $\hat{\bH}$ is used by agent $m=1$. Let us fix
the initial wealth and the strategies of the other agents, and pass on to
a realization of the strategies $h_t^m = (\alpha_t^m, \beta_t^m)$, wealth
$v_t^m$, and relative wealth $r_t^m$ as functions of $\omega$ only. In
notation for agent 1, we will also use the hat instead of the superscript
``1'', i.e.\ $\hat\alpha = \alpha^1$, $\hat\beta=\beta^1$, etc.

Introduce the predictable sequence of random vectors $F_t\in \R_+^{N_2}$
with the components
\[
F_t^{n} = \frac{\hat\beta_t^{n}}{\sum_m r_{t}^m\beta_t^{m,n}},
\]
where $0/0=0$. From \eqref{wealth-equation}, we obtain the relations
\[
\hat v_{t+1} = \biggl(\s{\hat\alpha_t}{X_{t+1}} +
\frac{\s{F_t}{\tilde Y_{t+1}}}{W_{t}}\biggr) \hat v_{t}, \qquad
W_{t+1} = \biggl( \sum_{m=1}^M r_{t}^m\s{\alpha_t^m}{X_{t+1}} + 
\frac{|\tilde Y_{t+1}|}{W_{t}}\biggr)W_{t}.
\]
Consequently, we find $\ln \hat r_{t+1} - \ln \hat r_{t} =
f_t(X_{t+1},\tilde Y_{t+1})$, where $f_t = f_t(\omega,x,y)$ is the
$\F_{t}\otimes\B(\R^N)$-measurable function
\[
f_t(x,y) = \ln \biggl(\frac{\s{\hat \alpha_t}{x}W_{t} + \s{F_t}{y}}
{W_{t}\sum_m r_{t}^m\s{\alpha_t^m}{x} + |y|}\biggr)
\]
(the argument $\omega$ is omitted for brevity).

We need to show that $\E_{t} f_t(X_{t+1},\tilde Y_{t+1}) \ge 0$. Rewrite the
function $f_t(x,y)$ as
\begin{equation}
\label{f1-f2}
\begin{split}
f_t(x,y) &= \ln\biggl(\frac{\s{\hat\alpha_t}{x}W_{t} + |y|}{W_{t}\sum_m
r_{t}^m\s{\alpha_t^m}{x} + |y|}\biggr)  +
\ln\biggl(\frac{\s{\hat\alpha_t}{x}W_{t} + \s{F_t}y}{\s{\hat\alpha_t}{x}W_{t} + |y|}\biggr)\\
&:= f^{(1)}_t(x,y) + f^{(2)}_t(x,y).
\end{split}
\end{equation}
For the first term, we can use the inequality $\ln x \ge 1-x^{-1}$ and apply
\eqref{alpha-ineq}, which gives
\begin{equation}
\label{f1-ineq}
\E_{t} f_t^{(1)}(X_{t+1},\tilde Y_{t+1}) \ge \E_{t}
\frac{\s{\hat\alpha_t-\sum_m
r_{t}^m\alpha_t^m}{X_{t+1}}W_{t}}{\s{\hat\alpha_t}{X_{t+1}}W_{t} + |\tilde Y_{t+1}|} \ge \biggl\langle
e,\;{\hat\alpha_t - \sum_{m=1}^M r_{t}^m \alpha_t^m}\biggr\rangle.
\end{equation}
For the second term in \eqref{f1-f2}, we have
\begin{equation}
\label{f2-ineq}
\E_{t} f_t^{(2)}(X_{t+1},\tilde Y_{t+1}) \ge
\E_{t} \frac{\s{\ln F_t}{\tilde Y_{t+1}}}{\s{\hat\alpha_t}{X_{t+1}}W_{t} + |\tilde Y_{t+1}|}\ge
|\hat\beta_t| - \sum_{m=1}^M r_{t}^m |\beta_t^m|,
\end{equation}
where the first inequality follows from the concavity of the logarithm, and
the second one follows from that $\ln F_t = \ln\hat\beta_t - \ln \sum_m
r_{t}^m \beta_t^m$ and inequality \eqref{beta-ineq}.

Using that $|\beta_t^m| + \s e{\alpha_t^m} = 1$, we see that $\E_{t}
f_t(X_{t+1},\tilde Y_{t+1}) \ge 0$, hence $\E_{t} \ln \hat r_{t+1} \ge \ln
\hat r_{t}$. Since $\ln\hat r_t$ is a non-positive sequence, this inequality
also implies the integrability of $\ln\hat r_t$ (by induction, beginning
with $\ln \hat r_0$), so it is a submartingale.
\end{proof}

\begin{proof}[Proof of Theorem~\ref{th-numeraire}]
When all the agents use $\hat\bH$, from~\eqref{prices} we find $\b p_t^n
= \bhb_t^n \b W_t$, and hence $\s{\bhb_t}{\b Z_{t+1}} = {\btY_{t+1}}/{\b
W_t}$. Adding \eqref{alpha-ineq} and \eqref{corollary}, we obtain
\eqref{numeraire}. Then \eqref{numeraire-2} follows by Jensen's inequality.
\end{proof}

\subsection{Proof of Theorem~\ref{th-dominance}}
We will need the following proposition which provides two
inequalities of a general nature.
\begin{proposition}
\label{lemma-ineqs}
1) For any $a,b\in (0,1]$ 
\begin{equation}
\ln\frac{a+b}{2} - \frac{\ln a + \ln b}{2} \ge
\frac{(a-b)^2}8.\label{log-ineq}
\end{equation}
2) Suppose
$x,y\in\R^N_+$ are two vectors such that $|x|\le1$, $|y|\le 1$, and
for each $n$ it holds that if $y^n=0$, then also $x^n=0$. Then
\begin{equation}
\s x{\ln x - \ln y} \ge \frac{\|x-y\|^2}{4} + |x|-|y|.\label{logsum}
\end{equation}
\end{proposition}
\begin{proof}
1) Assume $a\le b$. The inequality clearly holds if $a=b$. Let $f(a)$ be the
difference of its left-hand side and right-hand side, with $b$ fixed. It is
enough to show that $f'(a)\le 0$ for $a\in(0,b]$. After differentiation,
this becomes equivalent to $a(a+b) \le 2$. The latter inequality is clearly
true, provided that $a,b\in(0,1]$.

2) Inequality \eqref{logsum} follows from a known inequality for the
Kullback-Leibler divergence if $x/|x|$ and $y/|y|$ are considered as
probability distributions on a set of $N$ elements. Its short direct proof can be
found in \citet[Lemma~2]{DrokinZhitlukhin20}.
\end{proof}

\begin{proof}[Proof of Theorem~\ref{th-dominance}]
We will use the same notation for realizations of strategies as in the proof
of Theorem~\ref{theorem}. It was shown that $\ln \hat r_t$ is a
submartingale. Let $c_t$ be its compensator, i.e.\ the predictable
non-decreasing sequence such that $\ln \hat r_t - c_t$ is a martingale; in
the explicit form
\[
c_t = \sum_{s\le t}( \E_{s-1} \ln \hat r_s - \ln \hat r_{s-1}).
\]
As was shown in the proof of Theorem~\ref{theorem},
\[
c_{t+1} - c_{t} = \E_{t} f_t(X_{t+1}, Y_{t+1}) =
\E_{t} (f_t^{(1)}(X_{t+1}, Y_{t+1}) + f_t^{(2)}(X_{t+1}, Y_{t+1}))
\]
with $f^{(1)},f^{(2)}$ defined in \eqref{f1-f2}. Since $\ln \hat r_t$ is
non-positive and converges, we have $c_\infty < \infty$ with probability 1.
Let us consider again inequalities \eqref{f1-ineq}--\eqref{f2-ineq} and
strengthen them using Proposition~\ref{lemma-ineqs}. Fix $t\ge 1$ and let
\[
a = \frac{\s{\hat  \alpha_t}{X_{t+1}}+ | Y_{t+1}|/W_{t}}{Q_{t+1}}, \qquad
b = \frac{\s{\bar \alpha_t}{X_{t+1}} + | Y_{t+1}|/W_{t}}{Q_{t+1}}.
\]
Note that $a,b\in(0,1]$. Then
\begin{equation}
\label{f1-ineq-strong}
\begin{split}
\E_{t}f_t^{(1)}(X_{t+1}, Y_{t+1}) &= 2\E_{t}\biggl(\ln a - \frac{\ln a + \ln
b}{2}\biggr)
\\&\ge 2
\E_{t}\biggl( \ln \frac{a+b}{2} - \frac{\ln a + \ln b}{2}\biggr) +
\s e{\hat\alpha_t - \bar\alpha_t} \\&\ge
\biggl(\frac{\s{\hat\alpha - \bar\alpha}{X_{t+1}}}{2Q_{t+1}}\biggr)^2
+ \s e{\hat\alpha_t - \bar\alpha_t}.
\end{split}
\end{equation}
Here, in the first inequality we used the estimate
\[
\E_{t} \biggl(\ln a - \ln \frac{a+b}{2}\biggr) \ge \frac12
\E_{t}\frac{\s{\hat \alpha_t-
\bar \alpha_t}{X_{t+1}}W_{t}}{\s{\hat \alpha_t}{X_{t+1}}W_{t} + | Y_{t+1}|}
\ge \frac{1}{2} \s{e}{\hat \alpha_t - \bar \alpha_t},
\]
which is obtained similarly to \eqref{f1-ineq}. In the second inequality of
\eqref{f1-ineq-strong} we applied \eqref{log-ineq}.

For the function $f^{(2)}$, using that there are no portfolio constraints on
the endogenous assets, so $\bhb_t$ is given by
\eqref{beta-explicit}, we find
\begin{equation}
\label{f2-ineq-strong}
\begin{split}
\E_{t} f_t^{(2)}(X_{t+1},Y_{t+1}) &\ge \E_{t} \frac{\s{\ln
F_t}{Y_{t+1}}}{\s{\hat\alpha_t}{X_{t+1}}W_{t} + | Y_{t+1}|} = \s{\ln
F_t}{\hat\beta_t} = \s{\hat\beta_t}{\ln \hat\beta_t - \ln \bar \beta_t} \\&\ge
\frac{\|\hat\beta_t - \bar
\beta_t\|^2}{4} +  |\hat\beta_t| - |\bar \beta_t|,
\end{split}
\end{equation}
where the first inequality is obtained similarly to \eqref{f2-ineq}, and in
the second one we applied~\eqref{logsum}. Consequently, from
\eqref{f1-ineq-strong}, \eqref{f2-ineq-strong}, we obtain
\[
c_{t+1}-c_t \ge
\biggl(\frac{\s{\hat\alpha_t - \bar\alpha_t}{X_{t+1}}}{2Q_{t+1}}\biggr)^2
+
\frac{\|\hat\beta_t - \bar
\beta_t\|^2}{4}.
\]
From here, using that $c_\infty<\infty$, we get \eqref{proximity}. Moreover,
$\hat \alpha_t - \bar \alpha_t = (1-\hat r_t)(\hat \alpha_t - \tilde
\alpha_t)$ and $\hat \beta_t - \bar \beta_t = (1-\hat r_t)(\hat \beta_t -
\tilde \beta_t)$, so on the set \eqref{dominance} we necessarily have
$\lim_{t\to\infty}\hat r_t=1$.
\end{proof}

\small 
\setlength{\bibsep}{0.2em plus 0.3em}
\bibliographystyle{apalike}
\bibliography{cgs}

\end{document}